\documentclass[a4paper]{article}

\usepackage{amsmath,amssymb,amsfonts,theorem}
\usepackage{graphicx,color}
\usepackage{subfigure}
\usepackage{authblk}
\usepackage{fullpage}
\usepackage{tikz}
\graphicspath{{./Figures/}}

{ \theorembodyfont{\normalfont} 
\newtheorem{example}{Example}
\newtheorem{remark}{Remark}
}
\newtheorem{assumption}{Assumption}
\newtheorem{definition}{Definition}
\newtheorem{theorem}{Theorem}

\newtheorem{proposition}{Proposition}

\newcommand{\R}{\mathbb{R}}

\newcommand{\1}{\mathbf{1}} 

\newcommand{\dist}{\rm dist}

\newcommand{\dst}{\displaystyle}

\def\be{\begin{equation}}
\def\ee{\end{equation}}
\def\ba{\begin{array}}
\def\ea{\end{array}}
\def\eqa{\begin{eqnarray}}
\def\eqe{\end{eqnarray}}

\normalsize

\newcommand{\mc}{\mathcal}

\begin{document}

\title{Internal models for nonlinear output agreement and optimal flow control
}
\author[1]{Mathias B\"{u}rger\thanks{Work supported in part by the Cluster of Excellence in Simulation Technology (EXC301/2) at the University of Stuttgart.}}
\author[2]{Claudio De Persis\thanks{Work supported by the research grants
{\it Efficient Distribution of Green Energy} (Danish Research Council of Strategic Research), {\it Flexiheat} (Ministerie van Economische Zaken, Landbouw en Innovatie), and by a starting grant of the Faculty of Mathematics and Natural Sciences, University of Groningen.}}
\affil[1]{Institute for Systems
Theory and Automatic Control, University of Stuttgart, Pfaffenwaldring 9, 70550 Stuttgart, Germany {\tt mathias.buerger@ist.uni-stuttgart.de}
}
\affil[2]{ITM, Faculty of Mathematics and Natural Sciences, University of Groningen, Nijenborgh 4,  9747 AG Groningen, the Netherlands {\tt c.de.persis@rug.nl}
}

\maketitle 

\begin{abstract}
This paper studies the problem of output agreement in networks of nonlinear dynamical systems under time-varying disturbances. Necessary and sufficient conditions for output agreement are derived for the class of incrementally passive systems. Following this, it is shown that the optimal distribution problem in dynamic inventory systems with time-varying supply and demand can be cast as a special version of the output agreement problem. We show in particular that the time-varying optimal distribution problem can be solved by applying an internal model controller to the dual variables of a certain convex network optimization problem. 
\end{abstract}


\section{Introduction}
%
Output agreement has evolved as one of the most important control objectives in cooperative control. It has been studied in various different contexts, ranging from distributed optimization (\cite{Tsitsiklis1986}) up to oscillator synchronization (\cite{Stan2007}). 
Adding up to these results, we discuss in this paper the output agreement problem in the context of optimal distribution control for inventory networks with time-varying supply. \\
%
Internal model control tools have been used to handle output agreement problems in a variety of different formulations, see e.g.,(\cite{Wieland2011}), (\cite{Bai2011}), (\cite{DePersis2012}). 
We consider here a different problem set-up, involving time-varying external disturbance signals, and solve the output agreement problem for the class of incrementally passive systems. Our derivations follow the trail opened by \cite{Pavlov2008}. \\
The output agreement problem with time-varying external signals, studied in this paper, turns out to be of particular relevance for the routing control in inventory systems. We consider a simple inventory system, taking into account the storage levels and routing between the different inventories. This dynamics models, e.g., supply chains (\cite{Alessandri2011}) or data networks (\cite{Moss1983}). A key challenge in such inventory systems is to handle a time-varying supply and demand in an optimal way, while using only a distributed control strategy.  \\
%
%
The contributions of this paper are twofold. 
First, we present necessary and sufficient conditions for the \emph{output agreement problem} under time-varying disturbances. We consider networks of nonlinear systems interacting according to an undirected network topology. Following an internal model control approach, we consider controllers placed on the edges of the network and provide necessary conditions for the output agreement problem to be feasible. Sufficient conditions for output agreement in networks of incrementally passive dynamical systems are provided. We prove that the output agreement problem is feasible if one can find an incrementally passive internal model controller.
As a second contribution, we show that the \emph{optimal distribution problem} in inventory systems with time-varying supply and demand can be cast as an output agreement problem. The necessary conditions for the optimal distribution problem are a specific representation of the well known regulator equations. Subsequently, we present controllers solving the optimal distribution problem, for either quadratic cost functions or constant supplies. The internal model approach generalizes the results of \cite{DePersis2013}. \\
%
The remainder of the paper is organized as follows. The basic problem formulation and necessary conditions for output agreement are presented in Section \ref{sec.OutputAgreementNec}. Sufficient conditions for output agreement in networks of incrementally passive systems are discussed in Section \ref{sec.Sufficient}. The time-varying optimal distribution problem is introduced in Section \ref{sec.OptimalDistribution}, where also necessary conditions are discussed. We present then the solution to the problem for linear-quadratic problems in Section \ref{sec.LQ} and for constant supplies in Section \ref{sec.ConstantSupply}.

\textbf{Preliminaries:}
The notation we employ is standard. The set of (positive) real numbers is denoted by $\mathbb{R}$ ($\mathbb{R}_{\geq}$). The distance of a point $q$ from a set $\mathcal{A}$ is defined as $\mbox{dist}_{\mc{A}} q = \mbox{inf}_{p \in \mc{A}} \|p - q\|$.
The range-space and null-space of a matrix $B$ are denoted by $\mc{R}(B)$ and $\mc{N}(B)$, respectively.
A graph $\mc{G}=(V,E)$ is an object consisting of a finite set of nodes, $|V| = n$, and edges, $|E| = m$.
The incidence matrix $B \in \mathbb{R}^{n \times m}$ of the
graph $\mc{G}$ with arbitrary orientation, is a $\{0, \pm 1\}$ matrix with $[B]_{ik}$ having value
`+1' if node $i$ is the initial node of edge $k$, `-1' if it is the terminal
node, and `0' otherwise. 
We refer sometimes to the flow space of $\mc{G}$ as the null space $\mc{N}(B^{T})$, and the cut space of $\mc{G}$ as the range space $\mc{R}(B)$. 
Additionally, $\mc{N}(B)$ is named the circulation space of $\mc{G}$, and  $\mc{R}(B^{T})$ the differential space. 

\section{Problem formulation and necessary conditions} \label{sec.OutputAgreementNec}

We consider a network of dynamical systems defined on a connected, undirected graph $\mc{G}=(V,E)$. Each node represents a nonlinear system 
\be\label{nonl.systems.0}\ba{rcll}
\dot x_i &=& f_i(x_i, u_i, w_i)\\
y_i &=& h_i(x_i, w_i), & i=1,2,\ldots, n,
\ea\ee
where $x_i\in \R^{r_i}$ is the state, $u_i , y_i\in \R^{p}$ are the input and the output, respectively. Each system \eqref{nonl.systems.0} is driven by the time-varying signal $w_i\in \R^{q_i}$, representing, e.g.,~a disturbance. 
We assume that the exogenous signal $w_i$ is generated by a dynamical system of the form $\dot w_i =s_i(w_i)$.  In what follows it will be useful to require a gradient structure, i.e., $s_i(w_i)=\nabla \Sigma_i(w_i)$, with $\Sigma_i(w_i)$ a concave function. Then for all $w_i, w_i'$
\[
(w_i-w_i')^T (s_i(w_i)- s_i(w_i'))\le 0. 
\]
As an example of $\nabla \Sigma_i(w_i)$ consider the linear function with skew-symmetric matrix
\[
\nabla \Sigma_i(w_i)= S_i w_i,\quad S_i^T+S_i=0. 
\]
We stack together the signals $w_i$, for $i=1,2,\ldots, n$, and obtain  the vector $w$, which satisfies the equation $\dot w =s(w)$. In what follows, whenever we refer to the solutions of  $\dot w =s(w)$, we assume that the initial condition is chosen in a compact set ${\cal W}={\cal W}_1\times \ldots \times {\cal W}_n$. 
Similarly, let $x$, $u$, and $y$, be the stacked vectors of $x_{i}, u_{i},$ and $y_{i}$, for $i=1,\ldots,n$, respectively. Using this notation, the totality of all systems (\ref{nonl.systems.0}) is given by
\be\label{nonl.systems.tot}
\ba{rcl}
\dot x &=& f(x, u, w)\\
y &=& h(x,w)
\ea
\ee
and set ${\cal X}=\R^{r_1}\times \ldots \times \R^{r_n}$. \\
The control objective is to reach output agreement of all nodes in the network, independent of the exact representation of the time-varying external signals. 
Therefore, a controller of the form 
\be\label{nonl.im}
\ba{rcl}
\dot \xi_k &=& F_{k}(\xi_k, v_k)\\
\lambda_k &=& H_{k}(\xi_k), \quad k=1,2,\ldots, m,
\ea
\ee
with state $\xi_k\in \R^{\nu_k}$ and input $v_k\in \R^p$, is placed between any pair of neighboring nodes.
When stacked together, the controllers (\ref{nonl.im}) give raise to the overall controller
\be\label{im.big}
\ba{rcl}
\dot \xi &=& F(\xi,v)\\
\lambda &=& H(\xi),
\ea\ee
where $\xi \in \Xi=\R^{\nu_1}\times\ldots\times \R^{\nu_m}$. 
Throughout the paper the following interconnection structure between the plants, placed on the nodes of $\mc{G}$, and the controllers, placed on the edges of $\mc{G}$ is considered. A controller (\ref{nonl.im}), associated with edge $k$ connecting nodes $i,j$, has access to the measurement of the relative outputs $y_i-y_j$. In vector notation, the relative outputs of the systems are
\be\label{interc0}
z= (B\otimes I_p)^T y.
\ee
The controllers are then driven by the systems via the interconnection condition
\be\label{interc1}
v=-z,
\ee
where $v$ are the stacked inputs of the controllers.
Additionally, the output of a controller $\lambda_{k}$ influences its two incident systems. Thus, the stacked vector of controller's output $\lambda$ drives the process (\ref{nonl.systems.tot}) via the interconnection\footnote{
The interconnection structure \eqref{interc0}, \eqref{interc2} naturally represents a canonical structure for distributed control laws. This structure is often considered in the context of passivity-based cooperative control, see e.g., \cite{Arcak2007}, \cite{Bai2011},  \cite{Schaft2012}, \cite{DePersis2012a}, \cite{Burger2013}.}
\be\label{interc2}
u = (B \otimes I_p)\lambda. 
\ee 
We are now ready to formally introduce the output agreement problem. 
\begin{definition} \label{def.OutputAgreement}
The \emph{output agreement problem} is solvable for the process (\ref{nonl.systems.tot}) under the interconnection relations (\ref{interc0}), (\ref{interc1}), (\ref{interc2}), if there exists controllers  (\ref{im.big}), such that every solution originating from ${\cal W}\times {\cal X}\times \Xi$ is bounded and satisfies $\lim_{t\to \infty}~(B^T\otimes~I_p)~y(t)~=~\mathbf{0}$.  
\end{definition}
The first step is to investigate \emph{necessary conditions} for the output agreement problem to be solvable. 
The closed-loop system 
(\ref{nonl.systems.tot}), (\ref{interc0}), (\ref{interc1}), (\ref{interc2}),  (\ref{im.big})
can be written as
\be\label{overall}\ba{rcll} 
\dot w &=& s(w)\\
\dot x &=& f (x, (B\otimes I_p) H(\xi), w)\\
\dot \xi &=& F(\xi,-(B\otimes I_p)^T h (x , w)). 
\ea\ee
If the output agreement problem is solvable, then all the solutions to (\ref{overall}) starting from 
${\cal W}\times {\cal X}\times \Xi$ converge to the $\omega$-limit set  $\Omega({\cal W}\times {\cal X}\times \Xi)$. Notice that by boundedness of the solutions such a set is non-empty, compact and invariant. 
Furthermore, the $\omega$-limit set  $\Omega({\cal W}\times {\cal X}\times \Xi)$ must be a subset of the set of all pairs $(w,x)$ for which $(B\otimes I_p)^T h (x, w) = 0$.
Let now $(w, x^w, \xi^w)$ be a solution to 
(\ref{overall}) starting from $\Omega({\cal W}\times {\cal X}\times \Xi)$. By the invariance of the $\omega$-limit set, we have
\be\label{regulator.equations}\ba{rcl}
\dot x^w &=& f(x^w, u^w,w)\\
0&=& (B\otimes I_p)^T h (x^w, w)
\ea\ee
and 
\be \label{controller.constraints_O}
\ba{rcl}
\dot \xi^w &=&  F(\xi^w,0)\\
u^w &=& (B\otimes I_{p}) H(\xi^w). 
\ea\ee
We summarize the necessary condition as follows.
\begin{proposition}
If the output agreement problem is solvable, then, for every initial condition in $\Omega({\cal W}\times {\cal X}\times \Xi)$, there must exist solutions $(w, x^w, \xi^w)$ such that  the equations (\ref{regulator.equations}), (\ref{controller.constraints_O}) are satisfied. 
\end{proposition}
The constraint \eqref{regulator.equations} ensures that there exists a feed-forward control input $u^{w}$ that keeps the systems in output agreement. The second constraint \eqref{controller.constraints_O} ensures that the controller \eqref{im.big} is able to generate this feed-forward input signal. \\
The constraints (\ref{controller.constraints_O}) can be rewritten independently of the controller (\cite{Isidori2010}). 
 Let in the following $\lambda^{w}$ be some trajectory satisfying $u^{w} = (B\otimes I_{q})\lambda^{w}$, where $u^{w}$ is a solution to \eqref{regulator.equations}. Note that $\lambda^{w}$ is only then uniquely defined if the graph $\mc{G}$ has no cycles. Otherwise, it can be varied in the circulation space of $\mc{G}$. 
Bearing in mind the structure of the controllers (\ref{im.big}), it descends from the constraints (\ref{controller.constraints_O}) that 
 \be\label{controller.constraints}
\ba{rcl}
\dot \xi^w &=&  F(\xi^w,0)\\
\lambda^w &=& H(\xi^w). 
\ea\ee
Suppose now, that there exists an integer $d$ and maps  $\tau: \mc{W} \mapsto \mathbb{R}^{d}$, $\phi : \mathbb{R}^{d} \mapsto \mathbb{R}^{d}$ and $\psi: \mathbb{R}^{d} \mapsto \mathbb{R}^{mp}$ satisfying 
\be\label{embedding}\ba{rcl}
\dst\frac{\partial \tau}{\partial w} s(w) &=& \phi(\tau(w))\\
\lambda^w &=& \psi(\tau(w)). 
\ea\ee
Observe that $\tau, \phi, \psi$ do not depend on the controller since $\lambda^w$ depends on $B$ and $u^w$, the latter being dependent only on the process to control and $B$ on the topology of the underlying graph. 
Now, the dynamical system
\be\label{internal.model}\ba{rcl}
\dot \eta &=& \phi(\eta), \; \eta \in \mathbb{R}^{d} \\
\lambda &=& \psi(\eta). 
\ea\ee
has the following property. If $\eta_0= \tau(w(0))$, then the solution $\eta(t)$ to (\ref{internal.model}) starting from $\eta_0$ is  such that $\lambda(t)=\lambda^w(t)$ for all $t\ge 0$. \\
%
For designing a controller with the structure \eqref{im.big}, i.e., that decomposes into controllers on the edges of $\mc{G}$,
 we introduce a vector $\eta_{k} \in \mathbb{R}^{d}$ for each edge $k = 1,\ldots,m,$ and denote with $\psi_{k}$ the entries of the vector valued function $\psi$ corresponding to the edge $k$. Each edge is now assigned a controller of the form
\be\label{internal.model.k1}\ba{rcll}
\dot \eta_k &=& \phi(\eta_k) \\
\lambda_k &=& \psi_k(\eta_k), & k=1,2,\ldots, m.  
\ea\ee
%
With the stacked vector $\eta = [\eta_{1}^{T}, \ldots, \eta_{m}^{T}]^{T}$ and the vector valued functions 
\begin{align} \label{eqn.GlobalOMFunc}
 \bar{\phi}(\eta) = \begin{bmatrix} \phi(\eta_{1}) \\ \vdots \\ \phi(\eta_{m})\end{bmatrix}, \quad  \bar{\psi}(\eta) = \begin{bmatrix} \psi_{1}(\eta_{1}) \\ \vdots \\ \psi_{m}(\eta_{m}) \end{bmatrix}.
 \end{align} 
the overall controller is given as
\begin{align} \label{sys.GlobalIM}
\begin{split}
\dot{\eta} &= \bar{\phi}(\eta) \\
\lambda &= \bar{\psi}(\eta).
\end{split}
\end{align}
Note that if the initial condition is chosen as $\eta_0= I_{m} \otimes \tau(w(0))$ then the solution $\eta(t)$ to (\ref{internal.model}) starting from $\eta_0$ is  such that $\lambda(t)=\lambda^w(t)$ for all $t\ge 0$,  which is the same property expressed by (\ref{controller.constraints}).

\begin{remark}
If the functions $\phi(\eta)$ and $\psi(\eta)$ are linear, that is $\phi(\eta) = \Phi\eta\; (\Phi \in \mathbb{R}^{d \times d})$, and $\psi(\eta) = \Psi\eta\; (\Psi \in \mathbb{R}^{mp \times d})$, then $\psi_{k} = \Psi_{k}^{T}\; (\Psi_{k}^{T} \in \mathbb{R}^{p \times d})$ are the rows of the matrix $\Psi$ corresponding to the edge $k$, and the global functions \eqref{eqn.GlobalOMFunc} are given by 
$ \bar{\phi}(\eta) = (I_{m} \otimes \Phi) \eta$, and $\bar{\psi}(\eta) = \mbox{block.diag}( \Psi_{1}^{T}, \ldots,  \Psi_{m}^{T}) \eta$ (see also \cite{DePersis2013}). 
\end{remark}
In summary, the overall controller (\ref{sys.GlobalIM}) can be interpreted as internal-model-based controllers placed at the edges of the graph $\mc{G}$, where all controllers have the same global internal model.
The role of the internal model in problems of coordination in networked systems has been investigated in \cite{Wieland2011} for linear systems and in \cite{Wieland2010a}, Chapter 5 for  nonlinear systems. The result above adds up to these results.

\begin{remark} \label{rmk.LinearSystems}
 The condition \eqref{embedding} might be difficult to meet for general nonlinear problems. However, the condition can in fact always be met if both the internal model and the desired feed-forward input are linear in the disturbance, i.e., $s(w)=Sw$ and $\lambda^w=\Lambda w$, see \cite{Francis1976}.
\end{remark}


\begin{remark}
Suppose that the $\omega$-limit set can be expressed as
$
\Omega({\cal W}\times {\cal X}\times \Xi)=\{(w, x, \xi)\,:\,  x=\pi(w), \xi=\pi_c(w)\}. 
$
Then $x^w= \pi(w)$ and the so-called regulator equations (\ref{regulator.equations}) express the existence of an invariant manifold where the ``regulation error" $(B^T\otimes I_p)y$ is identically zero provided that the control input $u^w$ is applied. The conditions (\ref{controller.constraints}) express the existence of a controller able to  provide $u^w$.  Moreover, (\ref{regulator.equations}), (\ref{controller.constraints}) take the familiar expressions, see e.g. \cite{Isidori1990}:
\begin{align}
\begin{split} \label{eqn.FBI_Sys}
\ba{rcl}
\dst\frac{\partial \pi}{\partial w} s(w) &=&  f(\pi(w), (B\otimes I_p) H(\pi_c(w)),w) \\
0&=& (B\otimes I_p)^T h (\pi(w), w)
\ea
\end{split}
\end{align}
and 
\begin{align}
\begin{split} \label{eqn.FBI_Cont}
\ba{rcl}
\dst\frac{\partial \pi_c}{\partial w} s(w) &=&  F(\pi_c(w),0)\\
\lambda(w) &=&H(\pi_c(w)).
\ea
\end{split}
\end{align}
\end{remark}

\setcounter{example}{0} 

\begin{example}
Consider linear systems of the form
\be\label{lin.sys}\ba{rcl}
\dot x_i &=& A_i x_i +G_i u_i  +P_i w_i\\
y_i &=& C_i x_i 
\ea\ee 
Define the stacked system with $x = [x_{1}^{T},\ldots, x_{n}^{T}]^{T},$ and $A=\textrm{block.diag}\{A_1, \ldots, A_N\}$. Let $u$, $w$, $G$, and $S$, be defined equivalently.
%
There exist $x^w=\Pi w$ and $u^w =\Gamma w$ satisfying (\ref{regulator.equations}) if and only if $\Pi, \Gamma$ satisfy the Sylvester equation
\begin{align}\label{re1}
\begin{split}
&\Pi S = A \Pi +G \Gamma +P \\ 
&(B\otimes I_{p})^T C\Pi=0.
\end{split}
\end{align}
The equations are feasible if and only if 
\[
\left(\ba{cc}
A-\lambda I & G\\
(B\otimes I_{p})^T  C & 0 
\ea
\right)
\]
is full-row rank for each $\lambda$ in the spectrum of $S$.%
\footnote{ We can establish at this point a connection to the equilibrium independent passivity framework studied in \cite{Burger2013}.
In the case $A_i$ is invertible and $S_i =0$, the matrix $\Pi_i$ that satisfies (\ref{re1}) for any $\Gamma_i$ is $\Pi_i =  -A_i^{-1}(B_i \Gamma_i +P_i)$ and $x_i^w =-A_i^{-1}(B_i \Gamma_i +P_i)w_i$. The output $y^w$ becomes $-C_i(A_i^{-1}(B_i \Gamma_i +P_i)w_i= -C_i A_i^{-1}B_i u_i^w -C_i A_i^{-1}P_i w_i$ which coincides with the equilibrium input-to-output maps $k_{y_i}(u_i)$ considered in \cite{Burger2013}, Example 3.4.} \\
\end{example}

%
%
%
%

\section{Output agreement under time-varying disturbances} \label{sec.Sufficient}
In this section we highlight sufficient conditions that lead to a solution of the problem for a special class of systems  (\ref{nonl.systems.0}), namely incrementally passive systems, see e.g., \cite{Pavlov2008}, to which we refer the reader for the definition of a \emph{regular} storage function.
\begin{definition}\label{pip}			
The system \eqref{nonl.systems.0} is said to be incrementally passive if there exists a $C^{1}$ regular storage function $V:\mathbb{R}_{\geq 0} \times \mathbb{R}^{r_{i}} \times \mathbb{R}^{r_{i}} \mapsto \mathbb{R}_{\geq 0}$ such that for any two inputs $u_{i}, u_{i}'$ and any two solutions $x_{i}$,$x_{i}'$, corresponding to these inputs, the respective outputs $y_{i}$, $y_{i}'$ satisfy
\be\label{di.x}\ba{l}
\dst\frac{\partial V_i}{\partial t}+\dst\frac{\partial V_i}{\partial x_i} f_i(x_i, u_i , w_i) +
\dst\frac{\partial V_i}{\partial x'_i} f_i(x_i', u_i ', w_i) \\
\le (y_i-y_i')^T (u_i -u_i ').
\ea
\ee
\end{definition}


\setcounter{example}{0} 
\begin{example}(ctnd.) Linear  systems of the form \eqref{lin.sys} that are passive from the input $u_i $ to the output $y_i$ satisfy the assumption above, with $V_i=\frac{1}{2} (x_i-x_i')^T Q_i (x_i-x_i')$ and  $Q_i=Q_i^T>0$ the matrix such that $A_i^T Q_i +Q_i A_i\le 0$ and $Q_i G_i =C_i^T$.  
\end{example}

\begin{example} 
Nonlinear systems of the form
\[\ba{rcl}
\dot x_i &=& f_i(x_i) +G_i u_i  +P_i w_i\\
y_i &=& C_i x_i 
\ea\] 
with $f_i(x_i)=\nabla F_i(x_i)$, $F_i(x_i)$ twice continuously differentiable and concave, and $G_i=C_i^T$ are incrementally passive. In fact, by concavity of $F_i(x_i)$, $(x_i-x_i')^T ( f_i(x_i)- f_i(x_i'))\le 0$, and $V_i=\frac{1}{2} (x_i-x_i')^T (x_i-x_i')$ is the incremental storage function. 
\end{example}
In the previous section, it was shown that the controllers at the edge have to take the form 
\be\label{internal.model.k}\ba{rcll}
\dot \eta_k &=& \phi(\eta_k) \\
\lambda_k &=& \psi_k(\eta_k), & k=1,2,\ldots, m. 
\ea\ee
%
Now, they must be completed by considering additional control inputs that guarantee the achievement of the steady state. While we require the internal model to be identical for all edges, i.e., $\phi(\eta_k)$, the new augmented systems might be different. 
Then, the controllers modify as 
\be\label{internal.model.k.plus.stab}\ba{rcll}
\dot \eta_k &=& \phi_{k}(\eta_k,v_k) \\
\lambda_k &=& \psi_k(\eta_k), & k=1,2,\ldots, m, 
\ea\ee
where all controllers reduce to the common internal model if no external forcing is applied, i.e., $\phi_k(\eta_k,0) =\phi(\eta_k) $. 
The following is the main standing assumption that the controllers must satisfy to solve the output agreement problem for the class of incrementally passive systems:

\begin{assumption}\label{a.ip.im}
For each $k=1,2,\ldots, m$, there exists  regular functions $W_k(\eta_k, \eta_k')$, with $W_k:\R^{q_k}\times \R^{q_k}\to \R_+$ such that
\begin{align}\label{di.eta} 
\begin{split}
 \dst\frac{\partial W_k}{\partial \eta_k} \phi_k(\eta_k, v_k) + 
 \dst\frac{\partial W_k}{\partial \eta'_k} \phi_k(\eta'_k, v'_k) \\
 \le
 (\lambda_k-\lambda_k')^T (v_k-v_k').
 \end{split}
\end{align}
\end{assumption}
It is in general difficult to design the incrementally passive 
 controllers above. A first simple example when the design is possible is when the feedforward control input is linear, that is (\ref{embedding}) is satisfied with $\tau= {\rm Id}$, $\phi=s$ and $\psi$ is a linear function of its argument. In this case, we let 
\[
\phi_k(\eta_k, 0)=s(\eta_k), \quad \psi_k(\eta_k)= M_k\eta_k
\]
and define 
\begin{align} \label{eqn.IncPassiveIM}
\phi_k(\eta_k, v_k)=s(\eta_k)+M_k^T v_k.
\end{align}
Then by definition of $s$ as the gradient of a concave function, the storage function $W_k(\eta_k, \eta_k')=\frac{1}{2} (\eta_k-\eta_k')^T (\eta_k-\eta_k')$ satisfies
\begin{align}
\begin{split} 
 \dst\frac{\partial W_k}{\partial \eta_k} \phi_k(\eta_k, v_k) +
 \dst\frac{\partial W_k}{\partial \eta'_k} \phi_k(\eta'_k, v'_k)
= \\
 (\eta_k-\eta_k')^T (s(w_k)- s(w_k)') \\
 +(\eta_k-\eta_k')^T M_k^T (v_k-v_k')\\
 \le   (\psi_k(\eta_k)-\psi_k(\eta_k')) (v_k-v_k'), 
\end{split}
\end{align}
that is (\ref{di.eta}).

%
%
%

%
We state below the main result of the section that, while  extending to networked systems the results 
of \cite{Pavlov2008}, provides a solution to the output agreement problem in the presence of time-varying disturbances. It is a slightly more general statement than Theorem 2 in \cite{DePersis2013}.  
\begin{theorem}\label{th1}
Consider the system (\ref{nonl.systems.tot})
\be\label{nonl.systems.tot.again}\ba{rcl}
\dot w &=& s(w)\\
\dot x &=& f(x,u,w)\\
\ea
\ee
and let the regulator equations (\ref{regulator.equations}) hold. Consider the controllers
\be\label{controller.nu}
\ba{rcl}
\dot \eta &=& \bar{\phi}(\eta,v)\\
\lambda &=& \bar{\psi}(\eta) +\nu
\ea\ee
where $\nu$ is an extra feedback to design, and let $\bar{\phi}$ and $\bar{\psi}$ be the stacked functions of $\phi_{k}(\eta_{k},v_{k})$ and $\psi_{k}(\eta_{k}$) with the internal model property being satisfied. Consider the interconnection conditions
\be\label{int.constr}\ba{ll}
u=(B\otimes I_p) \lambda, \\ 
v=-(B^T\otimes I_p) y. \\ 
\ea\ee
If Assumption \ref{a.ip.im} holds and 
\[
\nu=-z,
\]
with $z=(B^T\otimes I_p) y$, 
then the output agreement problem is solvable, that is every solution starting from ${\cal W}\times {\cal X}\times \Xi$ is bounded and 
\[
\lim_{t\to+\infty} z(t) = \lim_{t\to+\infty} (B\otimes I_p)^T y(t) =\mathbf{0}.  
\]
\end{theorem} 

\begin{proof}
By the incremental passivity property of  the $x$ subsystem in (\ref{nonl.systems.tot.again})  and (\ref{regulator.equations}), it is true that 
\begin{align*}
\dst\frac{\partial V}{\partial t}+\dst\frac{\partial V}{\partial x} f(x, u , w) +
\dst\frac{\partial V}{\partial x^w} f(x^w, u^w, w) \\
\le (y-y^w)^T (u -u^w),
\end{align*}
where $V=\sum_i V_i$. Similarly by Assumption \ref{a.ip.im}, the system  (\ref{controller.nu}) satisfies
\[
 \dst\frac{\partial W}{\partial \eta} \bar{\phi}(\eta, v) +
 \dst\frac{\partial W}{\partial \eta^w} \bar{\phi}(\eta^w) \le
 (\lambda-\lambda^w)^T v - (\nu - \nu^{w})v,
\]
with $W= \sum_k W_k$ and $\bar{\phi}(\eta^{w}) = I_{m} \otimes \phi(\eta^{w})$.
 Bearing in mind the interconnection constraints
$u=(B\otimes I_p)  \lambda$,  $u^w=(B\otimes I_p) \lambda^w,$ and $v=-(B^T\otimes I_p) y,$  
and letting $U((x,x^w), (\eta, \eta^w))=V(x,x^w)+W(\eta, \eta^w)$ we obtain 
\begin{align*}
&\dot U((x,x^w), (\eta, \eta^w)) := \dot{V}(x,x^w) + \dot{W}(\eta, \eta^w)\\
& = (y-y^w)^T (u -u^w) +(\lambda-\lambda^w)^T v - (\nu-\nu^{w})^{T}v \\
& = (y-y^w)^T (B \otimes I_p)(\lambda-\lambda^w) \\
&-(\lambda-\lambda^w)^T (B^T \otimes I_p)y+(\nu-\nu^{w})^T (B^T \otimes I_p)y.
\end{align*}
By definition of output agreeement, $(B \otimes I_p)^T y^w=0$ and the previous equality becomes
\[\ba{rcl}
\dot U((x,x^w), (\eta, \eta^w)) &= & 
 \nu^T (B^T \otimes I_p)y\\
&=&  -||(B^T \otimes I_p)y||^2= -z^T z,
\ea\]
by definition of $\nu=-z$ and $\nu^{w} = 0$. Since $U$ is non-negative and non-increasing, then $U(t)$ is bounded. As $x^w, \eta^w$ are bounded\footnote{By definition, $(w,x^w, \eta^w)$ belongs to the $\omega$-limit set, which is compact. Hence, $x^w, \eta^w$ are bounded.}  and $U$ is regular, then $x,\eta$ are bounded as well. Hence the solutions exist for all $t$.  
Integrating the latter inequality we obtain 
\[
\dst\int_0^{+\infty} z^T(s) z(s) ds \le U(0). 
\]
By Barbalat's lemma, if one proves that $\frac{d}{dt} z^T(t) z(t)$ is bounded then one can conclude that $z^T(t) z(t)  \to 0$. Now, $z(t)=(B^T\otimes I_p) y=(B^T\otimes I_p) h(x,w)$ is bounded because $x,w$ are bounded. If $h$ is continuously differentiable and $\dot x, \dot w$ are bounded, then $\dot z$ is bounded and one can infer that $\frac{d}{dt} z^T(t) z(t)$ is bounded. By assumption, $w$ is the solution of $\dot w=s(w)$ starting from a forward invariant compact set. Hence, both $w$ and $\dot w$ are bounded. On the other hand, $\dot x$ satisfies
\[
\dot x = f(x, (B\otimes I_p)\bar{\psi}(\eta)-z, w)
\]
which proves that  it is bounded because $x, \eta, z$ were proven to be bounded, while $w$ is bounded by assumption. Therefore, $\dot x, \dot w$ are bounded and this implies that $\frac{d}{dt} z^T(t) z(t)$ is bounded. Then by Barbalat's Lemma we have $\lim_{t\to+\infty} z(t)=\mathbf{0}$ as claimed. 
\end{proof}

\subsection{Linear systems and distribution networks} \label{sec.OALinear}

We investigate next the output agreement problem for linear dynamical systems and focus on a routing control problem in inventory systems under time-varying demand and supply. 
%
%
Consider an inventory system with $n$ inventories and $m$ transportation lines, and let $B$ be the incidence matrix of the transportation network. The dynamics of the inventory system is given as
\begin{align}
\begin{split} \label{sys.Inventory1}
\dot{x} &= B\lambda + Pw,
\end{split}
\end{align}
where $x \in \mathbb{R}^{n}$ represents the storage level, $\lambda \in \mathbb{R}^{m}$ the flow along one line, and $Pw$ an external in-/outflow of the inventories, i.e., the supply or demand. We assume that the time varying supply/demand is generated by a linear dynamics
\begin{align}\label{sys.Inventory_w} \dot{w} = Sw \end{align}
and that it is balanced at any time, i.e., $\1^{T}Pw(t) = 0$ for all $t \geq 0$.  
The control problem is to find a distributed control law of the form
\begin{align} \label{sys.RoutingController}
\begin{split}
 \eta_{k} &= \Phi_{k}\eta_{k} + \Lambda_{k}v_{k} \\
 \lambda_{k} &= \Psi_{k}\eta_{k} + \Upsilon_{k}v_{k}
 \end{split}
\end{align}
such that for each initial condition $(w_{0}, x_{0}, \eta_{0})$ the solution of the closed loop system remains bounded and a balancing of the inventory levels is achieved, i.e.,  $\lim_{t \rightarrow \infty} B^{T}x = 0$. 

The closed-loop system \eqref{sys.Inventory1}, \eqref{sys.Inventory_w}, \eqref{sys.RoutingController} can be understood as feedback interconnection \eqref{interc0}, \eqref{interc1}, \eqref{interc2}, of the controller \eqref{sys.RoutingController} and the linear system 
\begin{align}
\begin{split}
\dot{w} &= Sw \\
\dot{x} &= u + Pw,
\end{split}
\end{align}
i.e., with $A = \mathbf{0}$ and $G = I_{n}$. Thus, the distribution problem can be understood as an output agreement problem.
The regulator equations \eqref{regulator.equations} become 
\begin{align}
\begin{split}
\Pi S = \Gamma + P, \quad
0 = B^{T}\Pi.
\end{split}
\end{align}
These conditions are satisfied with $\Pi = J$, i.e. the all-ones-matrix of appropriate dimension, and $\Gamma = JS-P$. 
We can directly conclude from Remark \ref{rmk.LinearSystems}, that there exists a matrix $H$ such that $\lambda^{w} = Hw$.\footnote{How such a matrix $H$ can be found is discussed in (\cite{DePersis2013}), and we refer the interested reader to this reference.} \\
The steady-state solution is such that $x^{w} = \Pi w = Jw = x_{*}^{w}\1$. 
 Since, by assumption, $\1^{T}Pw = 0$ at any time instant, it holds that $\1^{T}\dot{x}(t) = 0$ for all $t \geq 0$. However, this implies that $\1^{T} \dot{x}^{w} = \dot{\beta} \1^{T}\1 = 0$. The latter condition can only be satisfied for $\dot{\beta} = 0$, and therefore the steady-state solution $x^{w}$ must be a constant vector. 
Now, the steady-state routing $\lambda^{w}$ satisfies
\begin{align}
0 = B \lambda^{w} + Pw.
\end{align} 
Thus, any controller \eqref{sys.RoutingController} solving the output agreement problem, solves at the same time the exact routing problem of the time-varying supply/demand. 
The condition \eqref{embedding} is satisfied with $\tau = Id$, $\phi = S$ and $\psi = H$. 
Following our previous discussion, it remains to design the incrementally passive local controllers. Let in the following $H_{k}^{T}$ denote the $k$-th row of the matrix $H$. We know from \eqref{eqn.IncPassiveIM}, that the internal model controller on the edge $k$ of the form
\begin{align}
\begin{split}
\dot{\eta}_{k} = S\eta_{k} + H_{k}v_{k} \\
\lambda_{k} = H_{k}^{T}\eta_{k}
\end{split}
\end{align}
is incrementally passive.
Finally, it follows directly from Theorem \ref{th1} that the distributed internal model controller
\begin{align}
\begin{split} \label{eqn.Inventory1_Routing}
\dot{\eta} &= \bar{S}\eta - \bar{H}B^{T}x \\
\lambda &= \bar{H}\eta - B^{T}x
\end{split}
\end{align}
where $\bar{S} = I_{n} \otimes S$ and $\bar{H} = \mbox{block.diag}(H_{1}^{T}, \ldots, H_{m}^{T})$,
solves the output agreement problem and, additionally, achieves an exact routing of the time-varying supply through the network. This result has been also proven in \cite{DePersis2013}, but is presented here in the more general context of output agreement.

%
%
%
%

%
%
\subsection{Output agreement in the case of constant disturbances}

The critical assumption in the derivation presented above seems the incremental passivity of the internal model, i.e., Assumption \ref{a.ip.im}. 
However, the proof above exploits Assumption \ref{a.ip.im} only in a weaker form. In particular, it requires the incremental passivity property (\ref{di.eta}) not to hold with respect to any two trajectories, but only with respect to the real and the steady-state trajectory, i.e., with 
 $\eta_k'=\eta_k^w, v_k'=0, \lambda_k'=\lambda^w$. 
 Bearing in mind this observation, it is possible to find a storage function $W_k$ that fulfills (\ref{di.eta}) in the case of \emph{constant disturbances}. In fact in this case  $(\eta_k^w, \lambda_k^w)$  satisfy 
$
\lambda_k^w =\psi_k(\eta_k^w)$ for some constant $\eta^{w}$, i.e., 
$\dot \eta_k^w = 0$.

Let now $\psi_k$ be a strongly monotone function, and 
consider the following storage function (\cite{Jayawardhana2007}, \cite{Burger2013}):
\begin{align} \label{eqn.Bregman}
\begin{split}
W_k(\eta_k,\eta_k^w)=  \Psi_k(\eta_k)- \Psi_k(\eta_k^w) \\
+\nabla \Psi_k^T(\eta_k^w)(\eta_k-\eta_k^w),
\end{split}
\end{align}
where $\Psi_k : \R^q \to \R$ is a twice continuously differentiable function such that $\nabla \Psi_k(\eta_k)= \psi_k(\eta_k)$.\footnote{Note that $W_{k}$ is the Bregman distance between $\eta$ and $\eta_{w}$ for the function $\Psi$, \cite{Bregman1967}.} 
Now if $\psi_k$ is  monotone, $\Psi_k$ is convex and, by the global under-estimator property of the gradient \cite{Boyd2003}, we have
\[
\Psi_k(\eta_k)\ge  \Psi_k(\eta_k^w)+
\nabla \Psi_k^T(\eta_k^w)(\eta_k-\eta_k^w)
\]
for each $\eta_k, \eta_k^w$. If $\Psi_k$ is strictly convex, i.e., $\psi_{k}$ is strongly monotone, then the previous inequality holds if and only if $\eta_k=\eta_k^w$ (\cite{Boyd2003}), and $W_k$ is regular (\cite{Jayawardhana2007}). Hence, $W_k$ is a non-negative storage function if $\Psi_k$ is convex, and a positive regular storage function if $\Psi_k$ is strictly convex. Furthermore,
\begin{align*}
 \dst\frac{\partial W_k}{\partial \eta_k} \phi_k(\eta_k, v_k) =
 (\psi_k(\eta_k)-\psi_k(\eta_k^w))^T v_k. 
 \end{align*}
Hence, in the case of constant disturbances, Assumption \ref{a.ip.im} is always fulfilled and the output agreement problem is solved by the controllers
\be\label{controller.nu.constant.w}
\ba{rcl}
\dot \eta &=& v= -(B\otimes I_p)y\\
\lambda &=& \bar \psi(\eta) -(B^T\otimes I_p)y. 
\ea\ee
Control laws of the form \eqref{controller.nu.constant.w} have been studied in the context of network clustering in \cite{Burger2011}, \cite{Burger2011a} and in a more general network optimization framework in \cite{Burger2013}.


\section{Time-varying optimal distribution problems}\label{sec.OptimalDistribution}

We revisit now the distribution problem of inventory systems discussed in Section \ref{sec.OALinear}, i.e.,
\begin{align}
\dot{x} = B\lambda + Pw,
\end{align}
 with a time-varying external demand/supply. Let for now the supply/demand vectors be generated by a possibly nonlinear dynamics $\dot{w} = s(w).$
In contrast to Section \ref{sec.OALinear}, the control objective is not only to balance the inventory levels, but additionally to achieve an \emph{optimal} routing in the network. We associate to each transportation line a cost
 \[ \mc{P}_{k}(\lambda_{k}), \, k = 1,2,\ldots,m \]
  for the flow $\lambda_{k}$. Let $\mc{P}_{k}$ be convex and continuously differentiable.
  %
 We aim here to design routing controllers on the transportation lines of the form \eqref{internal.model.k},
taking only the imbalance between the incident inventories as inputs such that asymptotically all inventory levels are balanced and the flows in the network minimize the flow cost defined by the cost functions $\mc{P}_{k}$. \\
Before moving to the dynamic control problem, we briefly review the \emph{static optimal distribution problem}. Consider a fixed constant supply vector $w$. The (static) optimal distribution problem is to find a routing $\lambda^{w} \in \mathbb{R}^{m}$ such that
\begin{align}
\begin{split} \label{prob.OFP}
\lambda^{w} = \mbox{arg}\min \;& \sum_{k=1}^{m}\mc{P}_{k}(\lambda_{k}) \\
& 0 = B \lambda + Pw.
\end{split}
\end{align}
In the following, the notation $\mc{P}(\lambda) = \sum_{k=1}^{m} \mc{P}_{k}(\lambda_{k})$ will be used. The Lagrangian function associated to \eqref{prob.OFP} with the multiplier $\zeta$ is 
\[ \mc{L}(\lambda,v) = \mc{P}(\lambda) + \zeta^{T}(-B \lambda - P w).\]
From the Lagrangian, one obtains directly the optimality conditions (KKT--conditions). In particular, $(\lambda^{w},v^{w})$ is an optimal primal/dual solution pair to \eqref{prob.OFP} if the following nonlinear equations hold
\begin{align}
\begin{split} \label{eqn.OptimalityConditions}
\nabla \mc{P}(\lambda^{w}) - B^{T}\zeta^{w} &= 0 \\
B \lambda^{w} + Pw &= 0.
\end{split}
\end{align}
Note that the first condition simply implies $\nabla \mc{P}(\lambda^{w}) \in \mc{R}(B^{T})$ since $\zeta^{w}$ has no further constraints. 
We define the optimal routing/supply pairs as
\begin{align*}
\begin{split}
\Gamma = \{ (\lambda,w) \in \mathbb{R}^{m} \times \mc{W} \; : &\; \nabla \mc{P}(\lambda) \in \mc{R}(B^{T}), \\
&\; B \lambda + Pw = 0  \}.
\end{split}
\end{align*}
In particular, $(\lambda^{w},w) \in \Gamma$ if and only if $\lambda^{w}$ is an optimal solution to the static optimal distribution problem \eqref{prob.OFP} with the supply vector $w$.
%
%
The main difficulty associated with the set $\Gamma$ relates to the constraint $\nabla \mc{P}(\lambda) \in \mc{R}(B^{T})$.  This constraint can be avoided if the optimality conditions are expressed in terms of the \emph{dual solutions} $\zeta$. 
In what follows, we impose the condition that $\nabla \mc{P}$ is invertible.\footnote{Note that this is always given if $\mc{P}$ is strongly convex.}
The two optimality conditions \eqref{eqn.OptimalityConditions} can now be expressed as the following single nonlinear expression 
\[ B \nabla \mc{P}^{-1}(B^{T}\zeta^{w}) + Pw = 0.\] 
Bearing this in mind, we define the set of optimal dual solutions as
$
\Gamma_{D} = \{ (\zeta,w) \in \mathbb{R}^{n} \times \mc{W} : B \nabla \mc{P}^{-1}(B^{T}\zeta) + Pw = 0 \}.
$
We want to emphasize two properties of $\Gamma_{D}$. First, if $(\zeta^{w},w) \in \Gamma_{D}$ then $(\zeta^{w}+ c\1,w) \in \Gamma_{D}$ for any $c \in \mathbb{R}$. Second, $(\zeta^{w},w) \in \Gamma_{D}$ if any only if the corresponding routing strategy $\lambda^{w} = \nabla \mc{P}^{-1}(B^{T}\zeta^{w})$ satisfies $(\lambda^{w},w) \in \Gamma$. 
 We are now ready to formalize the dynamic problem. 
\begin{definition}
The \emph{time-varying optimal distribution problem} is solvable for the system \eqref{sys.Inventory1}, if there exists a controller \eqref{internal.model.k} such that any solution originating from $\mc{W} \times \mc{X} \times \Xi$ is bounded and satisfies
(i) $\lim_{t \rightarrow \infty} B^{T}x(t) = 0$; and
(ii)  $\lim_{t\rightarrow \infty} \mathrm{dist}_{\Gamma}\bigl( \lambda(t),w(t) \bigr) = 0.$ 
\end{definition}
Our previous discussion revealed that it can be advantageous to express the optimality conditions in terms of the dual solutions $\zeta$. In fact, we use this observation here to refine the control and restrict our attention to dynamic controllers of the form
\begin{align}
\begin{split} \label{sys.InventoryControllerDual}
\dot{\eta} &= \phi(\eta,v) \\
\lambda &= \nabla \mc{P}^{-1}(B^{T}\psi(\eta)). 
\end{split}
\end{align}
This control structure can be interpreted as an internal model controller for the \emph{dual} variables of the optimal network flow problem. 
For the time-varying optimal distribution problem to be feasible, it is necessary that the manifold
\begin{align}
H(w,x,\eta) = \begin{bmatrix} B^{T}x \\ B\nabla \mc{P}^{-1}(B^{T}\psi(\eta)) + Pw  \end{bmatrix} = 0
\end{align}
is invariant under the closed-loop dynamics
\begin{align} \label{sys.AutonomousFlow}
\begin{split}
\dot{w} &= s(w) \\
\dot{x} &= B\nabla \mc{P}^{-1}(B^{T}\psi(\eta))) + Pw \\
\dot{\eta} &= \phi(\eta,B^{T}x).
\end{split}
\end{align}
Note that, in contrast to the original output regulation problem, the ``output" function $H$ depends explicitly on the state of the controller $\eta$. 
However, at this point, the advantage of the internal model controller design for the dual variables becomes obvious. Let  $(w,x^{w},\eta^{w})$ be a solution 
to \eqref{sys.AutonomousFlow} starting in $\Omega(\mc{W} \times \mc{X} \times \Xi)$, satisfying 
\begin{align} \label{eqn.Balancing}
h(x) := B^{T}x^{w} = 0. 
\end{align}
The previous condition implies $x^{w}  = \beta \1$. Now, by the structure of the inventory dynamics follows $\1^{T}\dot{x} = 0$ at any time, and consequently
\[\dot{x}^{w} = 0 = B\nabla \mc{P}^{-1}(B^{T}\psi(\eta^{w})) ) + Pw.\]
Thus, the corresponding routing strategy
$ \lambda^{w} = \nabla \mc{P}^{-1}(B^{T}$ $\psi(\eta^{w})) ) $
is optimal at any point in time, i.e., $(\lambda^{w}(t), w(t))\in\Gamma$. Thus, by restricting the discussion to the ``dual" controller structure \eqref{sys.InventoryControllerDual}, we transformed the time-varying optimal distribution problem into an output agreement problem. \\
%
If the time-varying optimal distribution problem is solvable, then for any solution of the system $\dot{w} = s(w)$ there exists a solution to $\eta^{w}$ satisfying $\dot{\eta}^{w} = \phi(\eta^{w})$
such that
\begin{align} \label{eqn.Optimality_Etaw}
B\nabla \mc{P}^{-1}(B^{T}\psi(\eta^{w}(t))) ) + Pw(t)= 0.
\end{align}
%
Under sufficient smoothness conditions on  $\mc{P}^{-1}$ and on $s(x)$, there exists a differentiable trajectory $\eta^{w}$ satisfying \eqref{eqn.Optimality_Etaw}. 
Keeping this observation in mind, we can formalize conditions on the controller \eqref{sys.InventoryControllerDual}. 
Suppose there exists a map $\tau : \mc{W} \mapsto \Xi$ such that
\begin{align}
\begin{split} \label{eqn.Invariance}
\frac{\partial \tau}{\partial w} s(w) = \phi(\tau(w)), \; \quad \forall w \in \mc{W},\\
B\nabla \mc{P}^{-1}(B^{T}\psi(\tau(w))) + Pw = 0.
\end{split}
\end{align}
Then, the dynamical system $\dot{\eta} = \phi(\eta), \mu = \psi(\eta)$ has the following properties. If $\eta_{0} = \tau(w(0))$, then the solution $\eta(t)$ starting at $\eta_{0}$ is such that $\psi(\eta(t)) = \zeta^{w}$, that is $(\psi(t),w(t)) \in \Gamma_{D}$. Thus, if the condition \eqref{eqn.Invariance} is met, the internal model controller generates the correct time-varying dual solutions, and can produce the optimal routing strategy. We formalize this in the following result.
\begin{proposition}
Let condition \eqref{eqn.Invariance} hold. Then, for any initial condition $w(0) \in \mc{W}$, there exist an initial condition $(x_{0}, \eta_{0}) \in \mc{X} \times \Xi$ such that the solutions $(w(t), x(t),\eta(t))$ to \eqref{sys.AutonomousFlow} satisfy
(i) $B^{T}x(t) = 0$ and
(ii) $(\lambda(t), w(t)) \in \Gamma$, where $\lambda(t) = \nabla \mc{P}^{-1}(B^{T}\psi(\eta(t)))$.
\end{proposition}
Condition \eqref{eqn.Invariance} is the counterpart to the original controller regulator equation \eqref{embedding}. Note that, while the regulator equations for the systems dynamics are easily satisfied, as we discussed in Section \ref{sec.OALinear}, the difficulty of the time-varying optimal distribution problem results from partially fixed structure of the controller \eqref{sys.InventoryControllerDual}. 
%
We discuss next the distributed controller design for two specific representations of the optimal distribution problem.

\subsection{The Linear-Quadratic Case} \label{sec.LQ}

Suppose the supply is generated by a linear system
\begin{align*}
\dot{w} = Sw,\quad S + S^{T} = 0,
\end{align*}
and the transportation cost functions are quadratic
\begin{align*}
\mc{P}(\lambda) = \frac{1}{2}\lambda^{T}Q\lambda,
\end{align*}
for $Q = \mbox{diag}(q_{1},\ldots,q_{m})$ and $q_{k} >0$. In the linear-quadratic case one can solve the conditions \eqref{eqn.Invariance} directly.
 Consider the conditions \eqref{eqn.Invariance} with $\tau(w) = w$, and $\psi(\tau(w)) = Hw$, i.e.,
\begin{align}
\begin{split} \label{eqn.LQ}
BQ^{-1}BHw + Pw &= 0 \\
Sw &= \phi(w).
\end{split}
\end{align}
This requires $\phi(\eta) = S\eta$, i.e.,  the controller is a copy of the exosystem.
It remains to find a matrix $H$. Note that $L_{Q} = BQ^{-1}B^{T}$ is a weighted Laplacian matrix of the network. As $L_{Q}$ has one eigenvalue at zero, with the corresponding eigenvector $\1$, it is not invertible. 
However, one possible solution to \eqref{eqn.LQ} is 
\[H = -L_{Q}^{\dagger}P,\] 
where $L_{Q}^{\dagger}$ \emph{pseudo-inverse} of the weighted Laplacian $L_{Q}$, see e.g., \cite{Gutman2004}.\footnote{By the properties of $L_Q^\dagger$, it is promptly verified 
that 
\begin{align*}
\begin{split}
B Hw +Pw  = - B Q^{-1} B^T L_Q^\dagger Pw  +Pw  = \\[-0.8em]
- (I-\frac{\mathbf{1} \mathbf{1}^T}{n}) P w +P w = \mathbf{0}.
\end{split}
\end{align*}
}  
We can now construct a distributed controller. 
Recall that \eqref{sys.InventoryControllerDual} provides an internal model controller for the dual variables, assigned to nodes of the network. Thus, we assign an internal model controller to each node and introduce the variable $\eta_{i} \in \mathbb{R}^{d}$, satisfying the dynamics
\begin{align}
\begin{split}
\dot{\eta}_{i} &= S\eta_{i} \\
\zeta_{i} &= H_{i}^{T}\eta_{i}, \; i = 1,\ldots,n,
\end{split}
\end{align}
where $H_{i}^{T}$ is the $i$-th row of $H$. The routing at one edge computes then simply as $\lambda_{k} = q_{k}^{-1}(\zeta_{j} -\zeta_{i})$, where $j,i$ are the nodes incident to edge $k$. 
After introducing $\eta$ and $\zeta$ as the stacked vectors of $\eta_{i}$ and $\zeta_{i}$, respectively, and the matrices $\bar{S} = (I_{n} \otimes S)$, $\bar{H} = \mbox{block.diag}(H_{1}^{T}, \ldots,H_{n}^{T})$, we can express the overall controller as
\begin{align}
\begin{split}
\dot{\eta} &= \bar{S}\eta \\
\lambda &= Q^{-1}B^{T}\zeta = Q^{-1}B^{T}\bar{H}\eta.
\end{split}
\end{align}
This distributed internal model controller needs to be augmented with additional control inputs to ensure convergence to the desired stead state. This is formalized in the next result.
\begin{theorem}
Consider the inventory system \eqref{sys.Inventory1} with the supply generated by the linear dynamics $\dot{w} = Sw$.
Consider the controller
\[\ba{rcl}
\dot \eta &=& \bar{S} \eta + \bar{H}^{T} B Q^{-1} \nu\\
\lambda  &=& Q^{-1}B^{T} \bar{H} \eta +\nu.
\ea\]
with the interconnection condition
$ \nu = -B^{T}x. $
Then, every solution of the closed-loop system is bounded and (i) $\lim_{t\to +\infty} B^{T}x = 0$, and (ii) $\lim_{t\to +\infty} {\dist_{\Gamma}}(\lambda(t), w(t)) = 0$. 
\end{theorem}
The proof of this result follows directly along the same lines as the proof of Theorem \ref{th1}, taking into account that $V(x,x^{w}) = \frac{1}{2}\|x-x^{w}\|^{2}$ and $W = \frac{1}{2}\|\eta - \eta^{w}\|^{2}$ are regular incremental storage functions.
For concluding the optimality, one has to note that any trajectory in the set $\mc{E} = \{ (x,\eta) : B^{T}x = 0\}$, which is invariant under the dynamics, is such that $\dot{x} = 0$. Thus, any trajectory $\tilde{\eta}$ in $\mc{E}$ is such that the corresponding routing
\[ 
\tilde{\lambda} = 
Q^{-1}B^{T}\bar H \tilde{\eta} = \nabla \mc{P}^{-1}(B^{T}\bar H \tilde{\eta}) 
\]
satisfies 
$B \tilde{\lambda} + Pw = 0,$
i.e., the optimality conditions \eqref{eqn.OptimalityConditions}. Thus, any routing $\tilde{\lambda}$ in the invariant set $\mc{E}$ is an optimal routing.

%
%
\subsection{Optimal distribution with constant supply} \label{sec.ConstantSupply}
A second version of the optimal distribution problem, that can be solved by the internal model approach, relates to problems with constant supply and demand, i.e., $s(w) = 0$. In this case, we need no restriction to quadratic objective functions, but can consider general convex cost functions $\mc{P}_{k}$.
The conditions \eqref{eqn.Invariance} become
\begin{align*}
0 = \phi(\tau(w)) \\
B\nabla \mc{P}^{-1}(B^{T}\psi(\tau(w))) + Pw = 0,
\end{align*}
and they are solved with $\phi(w) = 0$ and $\tau(w)$ such that $(\tau(w),w) \in \Gamma_{D}$. 
In the static case, it becomes again advantageous to place the internal model controllers at the edges, and not at the network nodes. However, we still consider an internal model for the dual variables, but consider now the dual variables $\sigma = B^{T}\zeta$ instead of $\zeta$ as introduced in \eqref{eqn.OptimalityConditions}. 
Consider now the controller
\begin{align}
\begin{split} \label{sys.IMC_Static}
\dot{\sigma} &= \nu, \quad \sigma(0) \in \mc{R}(B^{T}) \\
\lambda &= \nabla \mc{P}^{-1}(\sigma) + \nu.
\end{split}
\end{align}
Note that if $\nu(t) \in \mc{R}(B^{T})$ then also $\sigma(t) \in \mc{R}(B^{T})$ for all times $t \geq 0$. If the initial condition $\sigma(0)$ is chosen as $\sigma^{w} = B^{T}\zeta^{w}$, with $(\zeta^{w},w) \in \Gamma_{D}$, then, under the input $\nu = 0$, the controller \eqref{sys.IMC_Static} generates the desired output. 
The internal model controller \eqref{sys.IMC_Static} is now not incrementally passive for arbitrary input signals, but it is incrementally passive with respect to any constant input signal. A storage function can therefore be found in the structure \eqref{eqn.Bregman}.
\footnote{The storage function here takes the form $W_{k} = \mc{P}^{\star}(\sigma) - \mc{P}^{\star}(\sigma^{w}) - \nabla \mc{P}^{\star}(\sigma^{w}) (\sigma - \sigma^{w})$, where $\mc{P}^{\star}$ is the convex conjugate of $\mc{P}$, see \cite{Rockafellar1997}, \cite{Burger2013}.}
Based on our previous discussions, we can now immediately conclude that the controller \eqref{sys.IMC_Static} with the interconnection condition $\nu = -B^{T}x$ solves the optimal distribution problem.

\section{Conclusions}
We proposed an internal model control design approach for output agreement of incrementally passive nonlinear systems with time-varying external disturbances. Building upon theses results, we studied the optimal distribution problem in inventory systems with time-varying supply and demand. These problems can be cast as output agreement problems if their dual formulation is considered. We showed how the optimal distribution problem can be solved by  internal model controllers for the dual variables. The specific solution to the distribution problem is discussed for the linear-quadratic case and for problems with constant supplies. \\
In the case of constant disturbances and equilibrium independent passive systems, the paper \cite{Burger2013} characterizes the value on which the outputs agree and establishes insightful relations with a number of network optimization problems. In the case of time-varying disturbances, such a characterization is still elusive and represents an interesting future research avenue. In the same work, the role of bounded interactions is investigated, a topic that has not been addressed in this paper. Passivity has provided a natural setting to study cooperative control algorithms in the presence of delays. The effect of delays in problems such as those studied in this paper has  been so far neglected,  to the best of our knowledge. It would be interesting to fill in this gap.

{\footnotesize
 \bibliographystyle{alpha}
\bibliography{Literature}

\begin{thebibliography}{JOGC07}

\bibitem[AGT11]{Alessandri2011}
A.~Alessandri, M.~Gaggero, and F.~Tonelli.
\newblock Min-max and predictive control for the management of distribution in
  supply chains.
\newblock {\em IEEE Transactions on Control Systems Technology}, 19(5):1075 --
  1089, 2011.

\bibitem[Arc07]{Arcak2007}
M.~Arcak.
\newblock Passivity as a design tool for group coordination.
\newblock {\em IEEE Transactions on Automatic Control}, 52(8):1380--1390, 2007.

\bibitem[BAW11]{Bai2011}
H.~Bai, M.~Arcak, and J.~Wen.
\newblock {\em Cooperative control design: {A} systematic, passivity--based
  approach}.
\newblock Springer, New York, NY, 2011.

\bibitem[Bre67]{Bregman1967}
L.M. Bregman.
\newblock The relaxation method of finding the common point of convex sets and
  its application to the solution of problems in convex programming.
\newblock {\em USSR Computational Mathematics and Mathematical Physics},
  7(3):200 -- 217, 1967.

\bibitem[BV03]{Boyd2003}
S.~Boyd and L.~Vandenberghe.
\newblock {\em Convex Optimization}.
\newblock Cambridge University Press, Cambridge, 2003.

\bibitem[BZA11]{Burger2011a}
M.~B{\"u}rger, D.~Zelazo, and F.~Allg{\"o}wer.
\newblock Network clustering: {A} dynamical systems and saddle-point
  perspective.
\newblock In {\em Proc. of IEEE Conference on Decision and Control}, pages
  7825--7830, Orlando, Florida, Dec. 2011.

\bibitem[BZA13a]{Burger2013}
M.~B{\"u}rger, D.~Zelazo, and F.~Allg{\"o}wer.
\newblock Duality and network theory in passivity-based cooperative control.
\newblock {\em Automatica}, 2013.
\newblock {s}ubmitted Jan. 2013.

\bibitem[BZA13b]{Burger2011}
M.~B{\"u}rger, D.~Zelazo, and F.~Allg{\"o}wer.
\newblock Hierarchical clustering of dynamical networks using a saddle-point
  analysis.
\newblock {\em IEEE Transactions on Automatic Control}, 58(1):113 -- 124, 2013.

\bibitem[{De~}13]{DePersis2013}
C.~{De~Persis}.
\newblock Balancing time-varying demand-supply in distribution networks: an
  internal model approach.
\newblock In {\em Proc. of European Control Conference}, 2013.
\newblock {S}ubmitted.

\bibitem[DJ12a]{DePersis2012a}
C.~{De~Persis} and B.~Jayawardhana.
\newblock Coordination of passive systems under quantized measurements.
\newblock {\em SIAM Journal on Control and Optimization}, 50(6):3155 -- 3177,
  2012.

\bibitem[DJ12b]{DePersis2012}
C.~{De~Persis} and B.~Jayawardhana.
\newblock On the internal model principle in formation control and in output
  synchronization of nonlinear systems.
\newblock In {\em Proc. of IEEE Conference on Decision and Control}, pages
  4894--4899, 2012.

\bibitem[Fra76]{Francis1976}
B.~A. Francis.
\newblock The linear multivariable regulator problem.
\newblock {\em SIAM Journal on Control and Optimization}, 14:486--505, 1976.

\bibitem[GX04]{Gutman2004}
I.~Gutman and W.~Xiao.
\newblock Generalized inverse of the {L}aplacian matrix and some applications.
\newblock {\em Bulletin T.CXXIX de l'Academie serbe des sciences et des arts -
  2004}, 29:15 -- 23, 2004.

\bibitem[IB90]{Isidori1990}
A.~Isidori and C.~Byrnes.
\newblock Output regulation of nonlinear systems.
\newblock {\em IEEE Transactions on Automatic Control}, 35(2):131 --140, 1990.

\bibitem[IM10]{Isidori2010}
A.~Isidori and L.~Marconi.
\newblock Nonlinear output regulation.
\newblock In W.S. Levine, editor, {\em The Control Handbook}. CRC Press, 2010.

\bibitem[JOGC07]{Jayawardhana2007}
B.~Jayawardhana, R.~Ortega, E.~{Garcia-Canseco}, and F.~Castanos.
\newblock Passivity on nonlinear incremental systems: Application to {PI}
  stabilization of nonlinear {RLC} circuits.
\newblock {\em Systems and Control Letters}, 56:618 -- 622, 2007.

\bibitem[MS83]{Moss1983}
F.~H. Moss and A.~Segall.
\newblock Optimal control approach to dynamic routing in networks.
\newblock {\em IEEE Transactions on Automatic Control}, AC-27(2):329--339,
  1983.

\bibitem[PM08]{Pavlov2008}
A.~Pavlov and L.~Marconi.
\newblock Incremental passivity and output regulation.
\newblock {\em Systems and Control Letters}, 57:400 -- 409, 2008.

\bibitem[Roc97]{Rockafellar1997}
R.~Rockafellar.
\newblock {\em Convex Analysis}.
\newblock Princeton University Press, 1997.

\bibitem[SS07]{Stan2007}
G.-B. Stan and R.~Sepulchre.
\newblock Analysis of interconnected oscillators by dissipativity theory.
\newblock {\em IEEE Transactions on Automatic Control}, 52(2):256 -- 270, 2007.

\bibitem[TBA86]{Tsitsiklis1986}
J.~Tsitsiklis, D.~Bertsekas, and M.~Athans.
\newblock Distributed asynchronous deterministic and stochastic gradient
  optimization algorithms.
\newblock {\em IEEE Transactions on Automatic Control}, 31(9):803--812, 1986.

\bibitem[vdSM12]{Schaft2012}
A.~J. van~der Schaft and B.~M. Maschke.
\newblock Port-hamiltonian systems on graphs.
\newblock Sep. 2012.
\newblock arXiv:1107.2006v2 [math.OC].

\bibitem[Wie10]{Wieland2010a}
P.~Wieland.
\newblock {\em From Static to Dynamic Couplings in Consensus and
  Synchronization among Identical and Non-Identical Systems}.
\newblock PhD thesis, Universit{\"a}t Stuttgart, 2010.

\bibitem[WSA11]{Wieland2011}
P.~Wieland, R.~Sepulchre, and F.~Allg{\"o}wer.
\newblock An internal model principle is necessary and sufficient for linear
  output synchronization.
\newblock {\em Automatica}, 47:1068 -- 1074, 2011.

\end{thebibliography}
}

\end{document}